\documentclass[points]{article}
\usepackage{amsmath,amssymb}
\usepackage{tikz}
\usepackage{amssymb,amscd,amsmath,amsthm}
\usepackage{graphicx}
\usepackage{pdfsync}

\def\bh\mathbf{h}

\def\cb{{\mathcal B}}

\def\cw{{\mathcal W}}

\def\bh{{\mathbf h}}

\def\a{\alpha}
\def\b{\beta}

\def\tr{{\rm Tr}}
\def\L{\Lambda}
\def\G{\Gamma}

\def\ffi{\varphi}

\def\Tr{\mathrm{Tr}}
\def\<{\langle}
\def\>{\rangle}

\def\1{\mathbf{1}}

\def\cw{\cal W}

\def\cal{\mathcal}

\def\id{{\bf 1}\!\!{\rm I}}

\newtheorem{theorem}{Theorem}[section]

\newtheorem{corollary}[theorem]{Corollary}
\newtheorem{lemma}[theorem]{Lemma}
\newtheorem{definition}[theorem]{Definition}
\newtheorem{example}[theorem]{Example}
\newtheorem{remark}[theorem]{Remark}

\def\cb{{\mathcal B}}

\def\cw{{\mathcal W}}

\def\bh{{\mathbb H}}

\def\a{\alpha}
\def\b{\beta}

\def\Tr{{\rm Tr}}
\def\L{\Lambda}
\def\G{\Gamma}

\def\ffi{\varphi}

\def\Tr{\mathrm{Tr}}
\def\<{\langle}
\def\>{\rangle}

\def\1{\mathbf{1}}

\def\cw{\cal W}

\def\cal{\mathcal}

\def\bh{\mathbf{h}}

\def\id{{\bf 1}\!\!{\rm I}}

\begin{document}

\begin{center}
{\Large {\bf  Entropy of Quantum Markov states on Cayley trees}}\\[1cm]
\end{center}

\begin{center}
{\large {\sc Farrukh Mukhamedov$^*$}}\\[2mm]
\textit{ Department of Mathematical Sciences,\\
College of Science, United Arab Emirates University,  \\
P.O. Box 15551,Al Ain, Abu Dhabi, UAE}\\
E-mail: {\tt far75m@yandex.ru, \ farrukh.m@uaeu.ac.ae}\\
$^*$ Corresponding author.
\end{center}

\begin{center}
{\sc Abdessatar Souissi$^*$  }\\
\textit{
  Department of Accounting, College of Business Management\\
Qassim University, Buraydah, Saudi Arabia} \\
E-mail: {\tt a.souaissi@qu.edu.sa}\\
$^*$ Corresponding author.
\end{center}

\begin{abstract}
In this paper, we continue the investigation of quantum Markov states (QMS) and define their mean
entropies. Such entropies are explicitly computed under certain conditions. The present work takes a huge leap forward at tackling one of the most important open problems in quantum probability which concerns the calculations of mean entropies of quantum Markov fields. Moreover, it opens new perspective for the generalization of many interesting results related to the one dimensional quantum Markov states
and chains to multi-dimensional cases.
\end{abstract}

\section{Introduction }\label{intr}

It is well-known that the classical and quantum dynamical entropies are important tools in many areas
of mathematics and play a huge role in information sciences \cite{OhyaPetz} and note that entropies of classical Markov chains are just as significant in statistical mechanics [10]. Quantum dynamical entropy has been studied by Connes and Sormer\cite{CS75}, Connes, Narnhofer, Thirring \cite{CNT87} and many others. Note that recently, the quantum dynamical entropy and the quantum dynamical mutual entropy were defined by Ohya in terms of complexity \cite{OhyaPetz}.  The notion of quantum Markov chain (QMC) was formulated by means of the transition expectation introduced by Accardi \cite{[Ac74f]}. Therefore, in \cite{AOW97,OW19} dynamical entropy through a quantum Markov chain has been introduced and calculated for some simple models. In this similar fashion and direction a mutual entropy in quantum Markov chains scheme has been investigated in \cite{Su97}. Note the significance of these contributions as these calculations are vital in both physics and information theory.

On the other hand, the study of quantum many-body systems has lived an explosion of results. This is specifically true in the field of Tensor Networks. Recently, "Matrix Product States", and more generally
"Tensor Network States" have played a crucial role in the description of the whole quantum systems \cite{CV,Or}.
We point out that an interesting mathematical approach to quantum states on tensor networks has closely
tied up with QMC on tensor product of matrix algebras \cite{[Ac74f],GZ}. Therefore, QMC have found massive
applications in several research domains of quantum statistical physics and information \cite{fannes,F.N.W},\cite{[CiPe-Ga-Sc-Ver17]},\cite{[RoOs96]}.

As we have mentioned, it is very important to develop a method which enables to compute the entropy
for quantum systems ((\cite{CNT87,Park92}). In \cite{Park94,Park} the dynamical entropy was computed in the sense of Connes,
Narnhofer, and Thirring \cite{CNT87} of shift automorphism of QMC. Moreover, the dynamical
entropy was proved to be equal to the mean entropy. It is important to emphasize that mean entropies calculations of
QMC have been obtained in several papers \cite{Fid01,Ohn04,Petz94}. However, in all these investigates QMC were
considered over a one dimensional lattice.

On the other hand, it is crucial to know certain formulas for the mean entropies of the
Gibbs measures on each of a sequence of large finite graphs. When the graph is non-amenable, then such a calculation becomes very tricky \cite{Alp15}. Recently, in \cite{AP2018} the mentioned formula has been found for Gibbs measured defined over regular trees (in particularly, Cayley trees).

We point out that, in \cite{AOM,MBS161,MBS162,MR2004} we have started to investigate particular classes of
quantum Markov chains (QMC) associated to the Ising types models on the Cayley
trees. It turned out that the
above considered QMCs fall into a special class called
\textit{quantum Markov states (QMS)} (see \cite{MS19,MS20}). Furthermore, in \cite{FM,MS19,MS20} a description of
QMS has been carried out. It is worth to stress that the considered QMS had the
Markov property not only with respect to levels of the considered tree, but
also with regard to the interaction domain at each site, which has a finer structure, and through a family of suitable
quasi-conditional expectations which are \textit{localized}
\cite{MS19,MS21}. Such a localization property is essential for the integral decomposition of QMS, since it takes into
account the finer structure of conditional expectations and
filtration \cite{[AcFiMu07],[AcSouElG20]}.

In this paper, we continue to investigate such kind of QMS and define for these states their mean entropies and under some conditions we are able to explicitly compute them.
The present work is another step towards one of the most important open problems in quantum probability, which concerns the calculations of means entropies of quantum Markov fields.
Moreover, it opens a new perspective for the generalization of  many interesting results related to one dimensional quantum  Markov states and chains to multi-dimensional cases.

\section{Preliminaries}

Let $\Gamma^k_+ = (V,E)$ be a semi-infinite Cayley tree of order
$k\geq 1$ with the root $x^0$ (i.e. each vertex of $\Gamma^k_+$
has exactly $k+1$ edges, except for the root $x^0$, which has $k$
edges). Here $L$ is the set of vertices and $E$ is the set of
edges. The vertices $x$ and $y$ are called {\it nearest neighbors}
and they are denoted by $l=<x,y>$ if there exists an edge
connecting them. A collection of the pairs
$<x,x_1>,\dots,<x_{d-1},y>$ is called a {\it path} from the point
$x$ to the point $y$. The distance $d(x,y), x,y\in V$, on the
Cayley tree, is the length of the shortest path from $x$ to $y$.

Recall a coordinate structure in $\G^k_+$:  every vertex $x$
(except for $x^0$) of $\G^k_+$ has coordinates $(i_1,\dots,i_n)$,
here $i_m\in\{1,\dots,k\}$, $1\leq m\leq n$ and for the vertex
$x^0$ we put $(0)$.  Namely, the symbol $(0)$ constitutes level 0,
and the sites $(i_1,\dots,i_n)$ form level $n$ (i.e. $d(x^0,x)=n$)
of the lattice.


The hierarchical structure on the considered tree is expressed as follows.
\[
W_n = \{ x\in V \, : \, d(x,x_0) = n\}
\qquad \Lambda_n =
\bigcup_{k=0}^n W_k, \qquad  \L_{[n,m]}=\bigcup_{k=n}^mW_k, \
(n<m).
\]

For $x\in \Gamma^k_+$, $x=(i_1,\dots,i_n)$ (i.e. $x\in W_n$) define the set of its direct successors vertices by
$$
S(x) := \{y\in \Lambda_n^c:\quad <x,y> \}\subset W_{n+1}
$$
and its $k^{th}$ successors  on the rooted tree $\Gamma^k_+$  is defined by the induction as follows
\begin{eqnarray*}
S_1(x) &=& S(x)\\
  S_{k+1}(x) &=& S(S_k(x)),\, \,   \forall k\ge 1.
\end{eqnarray*}
Put
\begin{equation}\label{S_1n_S_infty}
T(x) = \bigcup_{k\ge0}S_{k}(x)
\end{equation}
where $S_0(x) = x.$
Denote
$$
\overrightarrow{S}(x) = \{(x,1), (x,2), \cdots, (x,k)\},
$$
here $(x,i)$ means that $(i_1,\dots,i_n,i)$.


Since the Cayley tree $\Gamma^k_+$  is regular and each vertex has exactly $k$ direct successors then according to the above enumeration, we define $k$ shifts on the tree as follows: for each $j\in \{1,2,\dots, k\}$ the $j^{th}$ shift is given by
\begin{equation}\label{shifts}
  \alpha_j(x) = (j, x)= (j, i_1, i_2, \cdots, i_n)\in W_{n+1}.
\end{equation}
for every $x=(i_1, i_2, \dots, i_n)\in W_n$. Let $g = (j_1, j_2, \cdots, j_N)\in V$ one defines
$$
\alpha_g(x) := \alpha_{j_1}\circ \alpha_{j_2}\circ\cdots\circ \alpha_{j_N}(x) = (j_1, j_2, \cdots, j_N, i_1, i_2,\cdots, i_n).
$$
The shift $\alpha_g$ maps the Cayley tree $\Gamma^k_+$ onto  its sub-tree $T_x$ defined by (\ref{S_1n_S_infty}).
It follows that $\alpha_x(0) = x$ and $\alpha_x(\Lambda_n)= S_n(x)$.

Let $\mathcal{A}$ be a $C^*$-algebra with unit $\id$. For each $x\in V$ we denote $\mathcal{A}_x$ the $C^*$-algebra which is the same as $\mathcal{A}$, i.e. $\mathcal{A}_x\equiv\mathcal{A}$. For each  $\Lambda\subset V$  we denote by $\mathcal{A}_\Lambda$ the algebra generated by $\{\mathcal{A}_x:\, x\in\Lambda\}$. In these notations, one has
$$
\mathcal{A}_{ \Lambda_{n}} = \bigotimes_{x\in\Lambda_n}\equiv  \mathcal{A}_{ \Lambda_{n}}\otimes\id_{W_{n+1}}\subset \mathcal{A}_{ \Lambda_{n+1}}.
$$
Consider the local $*-$algebra
\begin{equation}\label{AVloc}
  \mathcal{A}_{V;\, loc} := \bigcup_{n\in\mathbb{N}}\mathcal{A}_{\Lambda_n}
\end{equation}
One has
$$
\mathcal{A}_V := \overline{\mathcal{A}_{V;\, loc}}^{C^*}.
$$
By $\mathcal{S}(\mathcal{B})$ we denote the set of states on a C$^*$--algebra $\mathcal{B}$.

\section{Quantum Markov states  and chains on Trees}

Recall that a transition expectation is a completely positive identity preserving map from a C$^*$--algebra into its C$^*$-algebra.
\begin{definition}
Let $\mathcal{C}\subset\mathcal{B}\subset\mathcal{A}$ be a triplet of unitary  C$^*$--algebras. A map $E: \mathcal{A}\to \mathcal{B}$ is called quasi-conditional expectation w.r.t. the given triplet if it is completely positive identity preserving and satisfies the following property
\begin{equation}
    E(ca) = cE(a); \quad \forall a\in\mathcal{A}, \forall c\in\mathcal{C}.
\end{equation}
\end{definition}
The typical form of transitions expectations $E_{\Lambda_n}$ w.r.t. the triplet $\mathcal{A}_{\Lambda_{n-1}}\subset\mathcal{A}_{\Lambda_n}\subset\mathcal{A}_{\Lambda_{n+1}}$ that we are going to investigate in the have the following form
\begin{equation}\label{CE_TE_Ln}
{E}_{\Lambda_n} = id_{\Lambda_{n-1}}\otimes\mathcal{E}_{W_n}
\end{equation}
where
$\mathcal{E}_{W_n}$ is a transition expectation from $\mathcal{A}_{ [N,n+1] }\to \mathcal{A}_{W_n}$.
Note that $\mathcal{E}_{W_n}$ is assumed to be  localized in the sense of \cite{MS19}. i.e. for each $x\in W_n$ there exists a transition expectation $\mathcal{E}_x$ from $\mathcal{A}_{\{x\}\cup S(x)}$ into $\mathcal{A}_x$ such that
\begin{equation}\label{localizedTE}
    \mathcal{E}_{W_n} = \bigotimes_{x\in W_n}\mathcal{E}_x.
\end{equation}

The shifts $\alpha_j$  given by \eqref{shifts} acts on the algebra  the algebra $\mathcal{A}_V$ as follows:
\begin{equation}\label{shift_algebra}
\alpha_j\left( \bigotimes_{ x\in  \Lambda_{n}}a_x \right):=
 \id^{(o)}\otimes\bigotimes_{i=0}^{n} a_x^{(j,x)}.
\end{equation}
\begin{definition}\label{def_TInv_MS}
  A state $\varphi$ on $\mathcal{A}_{L}$ is said to be \textit{translation-invariant} if
 \begin{equation}\label{trans_inv_state}
   \varphi\circ \alpha_j = \varphi
 \end{equation}
for every $j \in \{1,2,\cdots, k\}$.
\end{definition}

\begin{remark}
The above definition generalizes the notion of translation-invariant states on the one-dimensional lattice which corresponds to the case $k=1$. Note that the entropy of quantum Markov chains on the quantum spin lattice has been studied in many papers (see  \cite{Park94, AOW97, Go98})
\end{remark}

 \begin{definition}\label{QMCdef} \cite{[AcSouElG20]} Let  $\phi_o\in \mathcal{S}(\mathcal{A}_o)$ be a (initial) state,
 $\{E_{ \Lambda_{n}}\}$ a sequence of quasi-conditional expectations w.r.t.
the triple $\mathcal{A}_{{\Lambda}_{n-1}}\subseteq \mathcal{A}_{{\Lambda}_{n}}\subseteq\mathcal{A}_{{\Lambda}_{n+1}} $ and
a sequence $h_{n}\in\mathcal{A}_{V_n, +}$ of boundary conditions such that for each $a\in \mathcal{A}_V$  the limit
\begin{equation}\label{lim_Mc}
\varphi(a): = \lim_{n\to\infty} \phi_0\circ E_{\Lambda_{0}}\circ
E_{ \Lambda_{1}} \circ \cdots \circ E_{\Lambda_{n}}(h_{n+1}^{1/2}ah_{n+1}^{1/2})
\end{equation}
exists in the weak-*-topology and defines a state.  Then the triplet $(\phi_o, (E_{ \Lambda_{n}})_n, (h_{n})_n)$ is called a quantum Markov chain on $\mathcal{A}_V$. In this case the limiting state $\varphi$ given by (\ref{lim_Mc})
is also called quantum Markov chain (QMC).
\end{definition}

\begin{definition}
A quantum Markov chain $\varphi$ on $\mathcal{A}_V$ is called \textit{quantum Markov state(QMS)} if one has
\begin{equation}\label{Markov_state_eq}
  \varphi\lceil_{\mathcal{A}_{\Lambda_n}}\circ E_{\Lambda_n} = \varphi\lceil_{\mathcal{A}_{\Lambda_{n+1}}}.
\end{equation}
 \end{definition}

Due to the contributions made in \cite{AFri83,GZ}, the conditional expectations can be used in place of quasi conditional expectations.
However, the boundary condition plays a more physically significant role in the existence of phase transitions (see for instance  \cite{AMSa1,MBS161, MBS162, Sou21}). In \cite{MS19, MS20, MS21}, the authors studied in detail the structure of quantum Markov states associated with localized transitions expectations on the Cayley tree.

 Let us now deal with the  entropy of states. Assume that $\varphi$ is a state on $\mathcal{A}_V$ and $\Lambda \subset V$ is a finite region.
 Let $S(\varphi\lceil_{\mathcal{A}_{\Lambda}}$) be the von Neumann entropy of the of the state $\varphi$ on the algebra $\mathcal{A}_\Lambda$ (see \cite{OhyaPetz}), where $\varphi\lceil_{\mathcal{A}_{\Lambda}}$ is the state on $\mathcal{A}_\Lambda$ defined by the restriction of $\varphi$ on $\mathcal{A}_\Lambda$. If $\mathcal{D}^{\varphi}_\Lambda$ is the density matrix of the state $\varphi\lceil\mathcal{A}_\Lambda$, then
 \begin{equation}
     S(\varphi\lceil_{\mathcal{A}_\Lambda}) = - \mathrm{Tr}(D_{\Lambda}^{\varphi}\log D_{\Lambda}^{\varphi})
 \end{equation}
 In the one dimensional case, where $V = \mathbb{N}$, using the subadditivity of $S(\varphi\lceil_{\mathcal{A}_\Lambda})$, it was shown that the limit
 \begin{equation}\label{mean_entropy1D}
     s(\varphi) := \lim_{n\to\infty}\frac{1}{n+1}S(\varphi\lceil_{\mathcal{A}_{[0,n]}})
 \end{equation}
 exists. The quantity $s(\varphi)$ is called the \textit{mean entropy} of the state $\varphi$ \cite{Fid01,Petz94}. In the next section, we are going to discuss such an entropy for quantum Markov states on the Cayley trees.


 \section{ Mean Entropy for quantum Markov states on Trees}

In this section, we introduce the mean entropy for quantum Markov states defined on trees.  Let $\varphi$ be a locally faithful quantum Markov state on $\mathcal{A}_V$, i.e.  $\varphi\lceil_{\mathcal{A}_{\Lambda_n}}$ is faithful for all $n\geq 1$. Indeed, by means of \cite{Ohn04} one can establish that if $\varphi$ is a locally faithful QMS, then $\varphi$ is faithful.
From the tree structure, it is quite natural to define a generalization of (\ref{mean_entropy1D}) as follows
     \begin{equation}\label{mean_entropy_tree}
     s(\varphi) := \lim_{n\to\infty}\frac{1}{|\Lambda_n|}S(\varphi\lceil_{\mathcal{A}_{\Lambda_n}})
 \end{equation}

Let $\{E_{\Lambda_n}\}_{n\ge 0}$ be the sequence conditional expectations w.r.t. the triplet $\mathcal{A}_{\Lambda_{n-1}}\subset \mathcal{A}_{\Lambda_{n}}\subset \mathcal{A}_{\Lambda_{n+1}}$ associated with $\varphi$. We consider the Radon-Nikodym derivatives $D^{\varphi}_{\Lambda_n}$ of the state $\varphi_{\Lambda_n}:= \varphi\lceil_{\mathcal{A}_{\Lambda_n}}$ with respect to the trace $\mathrm{Tr}_{\mathcal{A}_{\Lambda_n}}$  on $\mathcal{A}_{\Lambda_n}$, i.e.
$$
\varphi_{\Lambda_n}(\cdot ) = \mathrm{Tr}_{\mathcal{A}_{\Lambda_n}}(D^{\varphi}_{\Lambda_n} \cdot )
$$
The transition expectation $\mathcal{E}_{W_n}$ associated with $E_{\Lambda_n}$ through (\ref{CE_TE_Ln})  is localized in the sense of \cite{MS19, MS20}. Then it satisfies (\ref{localizedTE}) for certain localized transition expectations $\mathcal{E}_x$ from $\mathcal{A}_{\{x\}\cup S(x)}$ into $\mathcal{A}_x$ with  $x\in W_n$.

For $X\subseteq Y$  two finite  sub-sets  of $V$,  we denote $T^Y_X$  the map defined by linear extension of
$$
 T^{Y}_{X}(a_X\otimes a_{Y\setminus X}) = \Tr(a_{Y\setminus X})a_X
$$
 One has
$$
T^{Y}_{X}(D^{\varphi}_{Y})  =  D^{\varphi}_{X}
$$
According to \cite{Mo06} the map $T_X^Y$ is the unique Umegaki conditional  expectation  from $\mathcal{A}_Y$ to $\mathcal{A}_X$ with respect to the tracial state.

 Let $A, B$ and $C$ are mutually disjoint finite subsets of the vertex $V$. One can easily check that
 \begin{eqnarray}
 &&T^{A,B,C}_{A,B}\lceil_{\mathcal{A}_{B\cup C}}  = T^{B,C}_{B} \label{T1}\\
 && \mathcal{A}_B = \mathcal{A}_{A\cup B}\cap \mathcal{A}_{B\cap C}\label{T2}\\
 && T^{A,B,C}_{B,C}T^{A,B,C}_{A,B} = T^{A,B,C}_{A,B}T^{A,B,C}_{B,C} = T^{A,B,C}_{C}\label{T3}
\end{eqnarray}

For each $x\in W_n$, assume that $K_{\{x\}\cup S(x)}\in\mathcal{A}_{\{x\}\cup S(x)}$ is the conditional density amplitude of the transition  expectation $\mathcal{E}_x$, i.e.
$$
\mathcal{E}_x (\,  \cdot \, ) = T_{\{x\}}^{\{x\}, S(x)} \left(K_{\{x\}\cup S(x)}^{*}\, \cdot\,  K_{\{x\}\cup S(x)}  \right)
$$
Using (\ref{localizedTE}),  the conditional density amplitude of the transition expectation $\mathcal{E}_{W_n}$ has the following form
\begin{equation}\label{Knn+1=prodKxs(x)}
     K_{[n,n+1]} := \bigotimes_{x\in W_n}K_{\{x\}\cup S(x)} \in\mathcal{A}_{  \Lambda_{[n,n+1]}  }
\end{equation}

\begin{lemma}\label{lem_density}
 Let $\varphi$ be a QMS on $\mathcal{A}_V$ together with its sequence of quasi-conditional expectations $E_{\Lambda_n}$.  In the above notations, one has
    \begin{equation}\label{DnDn+1}
    D^\varphi_{\Lambda_{n+1}} = K_{ [n,n+1] }D^\varphi_{\Lambda_{n}} K_{ [n,n+1] }^{\ast}
    \end{equation}
\end{lemma}
\begin{proof} Let $a\in \mathcal{A}_{\Lambda_{n+1}}$, one has
\begin{eqnarray*}
\varphi\lceil_{\mathcal{A}_{\Lambda_n}}(E_{\Lambda_n}(a)) &=& \mathrm{Tr}_{\Lambda_{n}}\left(D^{\varphi}_{\Lambda_n}
\mathrm{T}^{W_n, W_{n+1}}_{W_n}\left(K_{ [N,n+1] }^{\ast}a K_{ [N,n+1] }\right)\right)\\
&=& \mathrm{Tr}_{\Lambda_{n+1}}\left(D^{\varphi}_{\Lambda_n}
K_{ [N,n+1] }^{\ast}a K_{ [N,n+1] }\right)\\
&=& \mathrm{Tr}_{\Lambda_{n+1}}\left(K_{ [N,n+1] }D^{\varphi}_{\Lambda_n}
K_{ [N,n+1] }^{\ast}a \right)
\end{eqnarray*}

Using (\ref{Markov_state_eq}), we arrive at the required identity.
\end{proof}
\begin{theorem}\label{thm_entropy1} Let $\varphi$ be a QMS on $\mathcal{A}_V$. Then
   \begin{equation}\label{entropy_dentity}
      S(\varphi \lceil_{\mathcal{A}_{\Lambda_{n+1}}}) +  S(\varphi \lceil_{\mathcal{A}_{W_{n}}}) = S(\varphi \lceil  \mathcal{A}_{\Lambda_{n}})+ S(\varphi\lceil_{\mathcal{A}_{  \Lambda_{[n,n+1]}  } }).
      \end{equation}
If in addition  $K_{ [n,n+1] }$ and $D^\varphi_{\Lambda_{n]}}$ commute, then
     \begin{equation}\label{S(phi_n+1)-S(phi_n)}
    S(\varphi \lceil_{\mathcal{A}_{\Lambda_{n+1}}})-  S(\varphi \lceil_{\mathcal{A}_{\Lambda_{n}}})   = -2 \sum_{x\in W_{n}}\varphi_{\mathcal{A}_{ \Lambda_{[n,n+1]} }}\left( \log |K_{\{x\}\cup \overrightarrow{S}(x)}|\right).
    \end{equation}
\end{theorem}
\begin{proof}
Since $\varphi$ is a Markov state on $\mathcal{A}_V$ with respect to the filtration $(\mathcal{A}_{\Lambda_n})_n$, for each $n$ there exists a quasi-conditional expectation $E_{n}$ with respect to the triplet  $\mathcal{A}_{\Lambda_{n-1}}\subset\mathcal{A}_{\Lambda_n}\subset\mathcal{A}_{\Lambda_{n+1}}$ satisfying
$$
\varphi\lceil_{\mathcal{A}_{\Lambda_n}}\circ T_{\Lambda_n}^{\Lambda_{n+1}}\circ E_n = \varphi\lceil_{\mathcal{A}_{\Lambda_n}}\circ E_n = \varphi\lceil_{\mathcal{A}_{\Lambda_{n+1}}}.
$$
The quasi-conditional expectation $E_n$ can be written as follows
$$
E_n = id_{\Lambda_{n-1}}\otimes E_{[n,n+1]}
$$
where $E_{[n,n+1]}$ is a transition expectation from $\mathcal{A}_{ \Lambda_{[n,n+1]} }$ into $\mathcal{A}_{W_n}$. It follows that
$$
\varphi\lceil_{\mathcal{A}_{W_n}}\circ T_{W_n}^{W_{n+1},W_{n+1}}\circ E_{[n, n+1]} = \varphi\lceil_{\mathcal{A}_{ \Lambda_{[n,n+1]} }}\circ E_n = \varphi\lceil_{\mathcal{A}_{ \Lambda_{[n,n+1]} }}\circ E_n = \varphi\lceil_{\mathcal{A}_{ \Lambda_{[n,n+1]} }}.
$$
Then in the notations of  (\ref{T1}), (\ref{T2}) and (\ref{T3}) for $A= W_{n+1}, B = W_n, C= \Lambda_{n-1}$,  the conditional expectation $T^{\Lambda_{n+1}}_{\Lambda_n}$ is enough (in the sense of \cite{Petez88, Mo06}) for the states $\varphi\lceil_{\mathcal{A}_{\Lambda_{n+1}}}$ and $\varphi\lceil_{\mathcal{A}_{ \Lambda_{[n,n+1]} }}$. Hence, by a result of \cite{Mo06}, we obtain
$$
S(\varphi\lceil_{\mathcal{A}_{\Lambda_{n+1}}})
- S(\varphi_{\mathcal{A}_{W_n, W_{n+1}}})  - S(\varphi\lceil_{\mathcal{A}_{\Lambda_{n}}}) + S(\varphi\lceil_{\mathcal{A}_{W_n}}) = 0
$$
This proves (\ref{entropy_dentity}).

Now assume that   $K_{ [N,n+1] }$ and $D^\varphi_{\Lambda_{n}}$ commute. Since  $T^{\Lambda_{n+1}}_{W_n, W_{n+1}}( D^\varphi_{\Lambda_{n}} ) = D^\varphi_{W_n}\otimes \id_{W_{n+1}}$,
by Lemma \ref{lem_density} the density matrices of $\varphi_{\Lambda_{n+1}}$ and $\varphi_{W_n, W_{n+1}}$ are $ D^\varphi_{\Lambda_n}|K_{[n,n+1]}|^2$ and $ D^\varphi_{\Lambda_[n,n+1]} = D^\varphi_{\Lambda_n} |K_{[n,n+1]}|^2$, respectively.
This leads to
\begin{eqnarray*}
  S(\varphi\lceil_{\mathcal{A}_{\Lambda_{n+1}}}) &=&    -\Tr\left( D^\varphi_{\Lambda_{n}} |K_{ [N,n+1] }|^2\log ( D^\varphi_{\Lambda_{n}} |K_{ [N,n+1] }|^2)\right) \\
  &=& -\Tr\left( D^\varphi_{\Lambda_{n}} \log  D^\varphi_{\Lambda_{n}}\mathcal{E}_{ [n,n+1] } (\id)\right)
  -2\Tr\left( D^\varphi_{\Lambda_{n}}|K_{[n,n+1]}|^2 \log  |K_{ [n,n+1] }|\right)\\
  &=& S(\varphi\lceil_{\mathcal{A}_{\Lambda_{n}}})
  -2\varphi\lceil_{\mathcal{A}_{{[n,n+1]}}}\left( \log|K_{{[n,n+1]}}| \right).
\end{eqnarray*}
On the other hand, we find
\begin{eqnarray*}
  S(\varphi\lceil_{\mathcal{A}_{\Lambda_{[n, n+1]}}}) &=& - \Tr( D^\varphi_{W_n} |K_{[n,n+1]}|^2\log(D^\varphi_{W_n} |K_{[n,n+1]}|^2)) \\
   &=&  - \Tr\left( D^\varphi_{W_n} \log\left(D^\varphi_{W_n}\right)\mathcal{E}_{ [N,n+1] }(\id)\right) - 2\Tr\left( D^\varphi_{W_n} |K_{[n,n+1]}|^2 \log |K_{[n,n+1]}| \right) \\
    &=&  S(\varphi\lceil  \mathcal{A}_{W_{ n}}) - 2\varphi\lceil_{\mathcal{A}_{ [N,n+1] }}\left( \log|K_{ [N,n+1] }| \right)\\
    &\overset{(\ref{Knn+1=prodKxs(x)})}{=}&  S(\varphi\lceil  \mathcal{A}_{W_{ n}})-2 \sum_{x\in W_{n}}\varphi_{\mathcal{A}_{\Lambda_{[n, n+1]}}}\left( \log |K_{\{x\}\cup \overrightarrow{S}(x)}|\right).
\end{eqnarray*}
which proves \eqref{S(phi_n+1)-S(phi_n)}.
\end{proof}

The following result consists a generalization to trees of a formula of the mean entropy for translation-invariant quantum Markov states on the quantum spin lattice \cite{OhyaPetz}.

\begin{corollary}\label{Coro_mean} Assume that conditions of Theorem \ref{thm_entropy1} are satisfied. If a QMS  $\varphi$ is translation-invariant then
\begin{equation}\label{mean_ent_formula}
\frac{S(\varphi_{\Lambda_{n+1}})-  S(\varphi_{\Lambda_{n}})}{|W_n| } =  S(\varphi_{\Lambda_{1]}})-  S(\varphi_{\Lambda_{0]}}).
\end{equation}
\end{corollary}
\begin{proof}
One has
$$
\varphi_{\Lambda_{n+1}}(\log |K_{\{x\}\cup \overrightarrow{S}(x)}|) =  \varphi_{\Lambda_{1]}}(\log |K_{\{o\}\cup \overrightarrow{S}(o)}|), \quad \forall x\in L.
$$
Then \eqref{S(phi_n+1)-S(phi_n)} leads to (\ref{mean_ent_formula}).
\end{proof}

It is known that,  the mean entropy (\ref{mean_entropy1D}) of a QMS on the one-dimensional spin lattice is obtained  by the following limit
$$
s(\varphi) = \lim_{n\to\infty} S(\varphi\lceil_{\mathcal{A}_{[0,n+1]}}) - S(\varphi\lceil_{\mathcal{A}_{[0,n]}})
$$
The following theorem proposes an extension of this result to the case of trees.

 \begin{theorem}
The mean entropy of the Quantum Markov state $\varphi$ satisfies
\begin{equation}\label{mean_entropy_form}
s(\varphi) = \lim_{n\to \infty}\frac{S(\varphi \lceil_{\mathcal{A}_{\Lambda_{n+1}}}) -S(\varphi \lceil_{\mathcal{A}_{\Lambda_{n}}}) }{|\Lambda_{n+1}|- |\Lambda_{n}| }
\end{equation}
If in addition, $ \varphi$ is translation-invariant then
\begin{equation}\label{m_entr_hom}
s(\varphi)  = \frac{S(\varphi \lceil_{\mathcal{A}_{\Lambda_{1}}}) -S(\varphi \lceil_{\mathcal{A}_{\Lambda_{0}}})}{k}.
\end{equation}
\end{theorem}

\begin{proof}
One has  
\begin{eqnarray*}
  \frac{S(\varphi \lceil  \mathcal{A}_{\Lambda_{n+1}}) -S(\varphi \lceil  \mathcal{A}_{\Lambda_{n}}) }{|\Lambda_{n+1}| -  |\Lambda_{n}|}
&=& \frac{|\Lambda_{n+1}|}{|\Lambda_{n+1}|}\frac{S(\varphi \lceil  \mathcal{A}_{\Lambda_{n+1}})}{|\Lambda_{n+1}|}-\frac{|\Lambda_{n}|}{|\Lambda_{n+1}|}
\frac{S(\varphi \lceil  \mathcal{A}_{\Lambda_{n}}) }
{|\Lambda_{n}|} \\
&=& \frac{k^{n+2}-1}{(k-1)k^{n+1}}\frac{S(\varphi \lceil  \mathcal{A}_{\Lambda_{n+1}})}{|\Lambda_{n+1}|}-\frac{k^{n+1}-1}{(k-1)k^{n}}
\frac{S(\varphi \lceil  \mathcal{A}_{\Lambda_{n}})}
{|\Lambda_{n}|}.
\end{eqnarray*}
This leads to  (\ref{mean_entropy_form}).
In the translation-invariant case,  by Corollary \ref{Coro_mean} we arrive at (\ref{m_entr_hom}).
\end{proof}

\begin{example} Let us consider the semi-infinite Cayley tree of order two $\Gamma^+_2$. Let $\mathcal{A}_{x}=M_{2}(\mathbb{C})$ for $x\in \mathcal{A}_{V}$. By $\sigma_{x}^{u}$,
$\sigma_{y}^{u}$, $\sigma_{z}^{u}$ we denote the standard Pauli spin matrices at site $u\in V$, i.e.
$$
\id^{(u)}=\begin{pmatrix}
1 & 0\\
0& 1
\end{pmatrix},\ \sigma_{x}^{(u)}=\begin{pmatrix}
0 & 1\\
 1& 0
\end{pmatrix} , \  \sigma_{y}^{(u)}=\begin{pmatrix}
0 & -i\\
 i& 0
\end{pmatrix}, \ \sigma_{z}^{(u)}=\begin{pmatrix}
1 & 0\\
 0& -1
\end{pmatrix}.
$$
For every vertices  $(u,(u,1),(u,2))$  we put
\begin{eqnarray}\label{1Kxy1}
&&K_{<u,(u,i)>}=\exp\{\b H_{u,(u,i)>}\}, \ \ i=1,2,\ \b>0,\\[2mm] \label{1Lxy1} &&
L_{>(u,1),(u,2)<}=\exp\{J\beta H_{>(u,1),(u,2)<}\}, \ \ J>0,
\end{eqnarray}
where
\begin{eqnarray}\label{1Hxy1}
&&
H_{<u,(u,i)>}=\frac{1}{2}\big(\id^{(u)}\id^{(u,i)}+\sigma^{(u)}_z\sigma^{(u,i)}_z\big),\\[2mm]
\label{1H>xy<1} &&
H_{>(u,1),(u,2)<}=\frac{1}{2}\big(\id^{(u,1)}\id^{(u,2)}+\sigma^{(u,1)}_z\sigma^{(u,2)}_z\big).
\end{eqnarray}

For each $n\in \mathbb{N}$ and $x\in V$, we denote
\begin{eqnarray}\label{K1}
&& A_{x,(x,1),(x,2)}=K_{<x,(x,1)>}K_{<x,(x,1)>}L_{>(x,1),(x,2)<},\\[2mm]
\label{K11} &&K_{[m,m+1]}:=\prod_{x\in \overrightarrow
W_{m}}A_{x,(x,1),(x,2)}, \ \ 1\leq m\leq n,\\[2mm]
 \label{K2} &&
\bh_{n}^{1/2}:=\prod_{x\in\overrightarrow
W_n}(h_{x})^{1/2}, \ \ \  \bh_n=\bh_{n}^{1/2}(\bh_{n}^{1/2})^*\\[2mm]
\label{K3}
&&{\mathbf{K}}_n:=\omega_{0}^{1/2}\prod_{m=1}^{n-1}K_{[m,m+1]}\bh_{n}^{1/2}\\[2mm]
\label{K4} && \mathcal{W}_{n]}:={\mathbf{K}}_n^{*}{\mathbf{K}}_n.
\end{eqnarray}
 Here $w_0\in\mathcal{A}_{o, +}$ is and initial state and $h_x\in\mathcal{A}_{x,+}$ is a family of positive boundary conditions.

Consider the following functional
$\ffi^{(n)}_{w_0,\bh}$, given by
\begin{eqnarray}\label{ffi-ff}
\varphi^{(n)}_{w_0,\bh}(a)=\tr(\cw_{n+1]}(a\otimes\id_{W_{n+1}})),
\end{eqnarray}
for every $a\in \cb_{\Lambda_n}$.
 In \cite{MBS161} we have shown that the compatibility equation (\ref{Markov_state_eq}) is satisfied for the sequence $\varphi^{(n)}_{w_0, \bh}$ under the following conditions:
 \begin{eqnarray}\label{eq1}
&&\Tr(\omega_{0}h_{0})=1, \\
\label{eq2} &&\Tr_{x]}\big({A_{x,(x,1),(x,2)}}h_{(x,1)}h_{(x,2)}
A_{x,(x,1),(x,2)}^{*}\big)=h_x, \ \ \textrm{for every}\ \ x\in L,
\end{eqnarray}
Moreover, it has been established the existence of three quantum Markov states $\varphi_1, \varphi_2$ and $\varphi_\alpha$ associated with three different translation-invariant solutions  $h_1, h_2$ and $h_\alpha$ of (\ref{eq1}) with
 $$
 \alpha =  \frac{4}{e^{2J\beta}(e^{4\beta}+1)+2e^{2\beta}}.
 $$
 In the following we will compute the mean entropy for the Markov state $\varphi_\alpha$. In \cite{MS201} we found the density  matrices $\mathcal{W}_{n]}^{(\alpha)}$ of $\varphi_\alpha$ as follows
 $$\mathcal{W}_{n]}^{(\varphi)} =  \exp\mathcal{H}_{n]}^{(\alpha)}$$
 \begin{equation}\label{densityphialpha}
 \mathcal{H}_{n]}^{(\alpha)}:=\sum_{x\in
W_{n-1]}}\mathcal{H}_{\alpha}^{(x,(x,1),(x,2))}
 \end{equation}
where
$$\mathcal{H}_{\alpha}^{(x,(x,1),(x,2))}:=\id^{\otimes |\{<_{(W_{n-1})}x\}|}\otimes\mathcal{H}^{(x)}_{\alpha}\otimes\id^{\otimes(|\{>_{(W_{n-1})}x\}|-2)},$$
here \small
$$\mathcal{H}_{\alpha}:= \ln\a\id_{M_8(\mathbb{C})}+
   2\left(
     \begin{array}{cccccccc}
       \beta(J+2) &  &  &  &  &  &  &  \\
         & \beta    &  &  &  &  &  &  \\
         &  & \beta    &  &  & \left(0\right) &  &  \\
         &  &  & \beta J &  &  &  &  \\
         &  &  &  & \beta J &  &  &  \\
         &  & \left(0\right) &  &  & \beta &  &  \\
         &  &  &  &  &  & \beta   &  \\
         &  &  &  &  &  &  & \beta(J+2)  \\
     \end{array}
   \right)$$
\normalsize and $\{<_{(W_{n-1})}x\}:=\{y\in W_{n-1}\mid y<x\}$,
$\{>_{(W_{n-1})}x\}:=\{y\in W_{n-1}\mid y>x\}$. In this notation,
the order $"x<y"$ is understood with respect to the lexicographic
order applied to the coordinates of $x$ and $y$.

It follows that
\begin{equation}\label{densityn}
     \mathcal{W}_{n]}^{(\varphi_\alpha)} =  \mathcal{W}_{n-1]}^{(\varphi_\alpha)}\exp\left(\sum_{x\in W_{n-1}}\mathcal{H}_{\alpha}^{(x,(x,1),(x,2))}\right)
\end{equation}
 The von Neuman entropy of $\varphi^{\alpha}\lceil_{\mathcal{A}_{\Lambda_n}}$ satisifies
 \begin{eqnarray*}
 S(\varphi_{\alpha}\lceil_{\mathcal{A}_{\Lambda_n}})
 &= & -\tr\left( \mathcal{W}_{n]}^{(\varphi_\alpha)}\log \mathcal{W}_{n]}^{(\varphi_\alpha)}\right)\\
 &\overset{(\ref{densityn})}{=}&  -\tr\left( \mathcal{W}_{n]}^{(\varphi_\alpha)}\left(\log \mathcal{W}_{n-1]}^{(\varphi_\alpha)} +  \sum_{x\in W_{n-1}}\mathcal{H}_{\alpha}^{(x,(x,1),(x,2))}  \right)\right) \\
&=&  -\tr\left( \mathcal{W}_{n]}^{(\varphi_\alpha)}\log \mathcal{W}_{n-1]}^{(\varphi_\alpha)}\right) -\sum_{x\in W_{n-1}}\tr\left( \mathcal{W}_{n]}^{(\varphi_\alpha)}   \mathcal{H}_{\alpha}^{(x,(x,1),(x,2))}   \right)\\
&=&  -\varphi_\alpha\lceil_{\mathcal{A}_{\Lambda_n}}\left(\log \mathcal{W}_{n-1]}^{(\varphi_\alpha)}\right) -\sum_{x\in W_{n-1}}\varphi_\alpha\lceil_{\mathcal{A}_{\Lambda_n}}\left(   \mathcal{H}_{\alpha}^{(x,(x,1),(x,2))}   \right) \\
&=&  -\varphi_\alpha\lceil_{\mathcal{A}_{\Lambda_{n-1}}}\left(\log \mathcal{W}_{n-1]}^{(\varphi_\alpha)}\right) -\sum_{x\in W_{n-1}}\varphi_\alpha\lceil_{\mathcal{A}_{\{x\}\cup S(x)}}\left(   \mathcal{H}_{\alpha}^{(x,(x,1),(x,2))}   \right)
 \end{eqnarray*}
 One has
 $$
  - \varphi_\alpha\lceil_{\mathcal{A}_{\Lambda_{n-1}}}\left(\log \mathcal{W}_{n-1]}^{(\varphi_\alpha)}\right)  =  -\tr\left( \mathcal{W}_{n-1]}^{(\varphi_\alpha)}\log \mathcal{W}_{n-1]}^{(\varphi_\alpha)}\right) = S(\varphi_{\alpha}\lceil_{\mathcal{A}_{\Lambda_{n-1}}})
 $$
 and since the state $\varphi_\alpha$ is translation-invariant then
 $$
 \varphi_\alpha\lceil_{\mathcal{A}_{\{x\}\cup S(x)}}\left(   \mathcal{H}_{\alpha}^{(x,(x,1),(x,2))}   \right)  = \varphi_\alpha\lceil_{\mathcal{A}_{\{x_0\}\cup S(x_0)}}\left(   \mathcal{H}_{\alpha}^{(x_0,(x_0,1),(x_0,2))}   \right)
 $$
 Therefore
 $$
S(\varphi_{\alpha}\lceil_{\mathcal{A}_{\Lambda_{n}}})  =
S(\varphi_{\alpha}\lceil_{\mathcal{A}_{\Lambda_{n-1}}})  - |W_{n-1}|\varphi_\alpha\lceil_{\mathcal{A}_{\{x_0\}\cup S(x_0)}}\left(   \mathcal{H}_{\alpha}^{(x_0,(x_0,1),(x_0,2))}   \right).
 $$
 This leads to (\ref{mean_ent_formula}). Using (\ref{m_entr_hom}) the mean entropy of the QMS $\varphi_\alpha$ is given by
 \begin{eqnarray*}
 s(\varphi_\alpha) &=& \frac{1}{2}\varphi_\alpha\lceil_{\mathcal{A}_{\{x_0\}\cup S(x_0)}}\left(   \mathcal{H}_{\alpha}^{(x_0,(x_0,1),(x_0,2))}   \right)\\
 &=& \frac{1}{2} \tr\left( \mathcal{H}_{\alpha}^{(x_0,(x_0,1),(x_0,2))} \exp\left( \mathcal{H}_{\alpha}^{(x_0,(x_0,1),(x_0,2))}\right)\right)\\
&=& \alpha\left( (\ln \alpha + 2\beta(J+2))e^{2\beta (J+2)} + 2 (\ln \alpha + 2\beta)e^{2\beta} + (\ln \alpha + 2\beta J)e^{2\beta J}\right)
 \end{eqnarray*}
We can compute the mean entropy of QMS $\varphi_\alpha$  directly from (\ref{mean_entropy_tree}). In fact,
\begin{eqnarray*}
S(\varphi_\alpha\lceil_{\mathcal{A}_{\Lambda_n}}) &=& - \tr\left(\mathcal{W}_{n]}^{(\varphi_\alpha)}\log \mathcal{W}_{n]}^{(\varphi_\alpha)} \right)\\
& \overset{(\ref{densityphialpha})}{=}&  -
\sum_{x\in \Lambda_{n-1}}
\tr\left(\mathcal{W}_{n]}^{(\varphi_\alpha)} \mathcal{H}_{\alpha}^{(x ,(x ,1),(x ,2))}\right)\\
&=&  -
\sum_{x\in \Lambda_{n-1}}
 \varphi_\alpha\lceil_{\mathcal{A}_{\{x\}\cup S(x)}}\left(  \mathcal{H}_{\alpha}^{(x ,(x,1),(x,2))}\right)\\
 &=&  -
|\Lambda_{n-1}|
 \varphi_\alpha\lceil_{\mathcal{A}_{\Lambda_1}}\left(  \mathcal{H}_{\alpha}^{(x_0,(x_0,1),(x_0,2))}\right)
\end{eqnarray*}
Then
\begin{eqnarray*}
s(\varphi_\alpha) &=& \lim_{n\to \infty}\frac{S(\varphi_\alpha\lceil_{\mathcal{A}_{\Lambda_n}})}
{|\Lambda_n|} = \lim_{n\to \infty}\frac{S(\varphi_\alpha\lceil_{\mathcal{A}_{\Lambda_n}})}
{|\Lambda_{n-1}|}\frac{|\Lambda_{n-1}|}{|\Lambda_n|} \\
&=& -\frac{1}{2}\tr\left( \mathcal{H}_{\alpha}^{(x_0,(x_0,1),(x_0,2))} \exp\left( \mathcal{H}_{\alpha}^{(x_0,(x_0,1),(x_0,2))}\right)\right).
\end{eqnarray*}
\end{example}

\section{Mixing property for quantum Markov states on  Cayley trees}

In this section, we consider the semi-infinite  Cayley tree $\Gamma^k_+$ and, it is assumed that $\mathcal{A}_x = M_d$  the $d\times d$ square matrices over the complex field $\mathbb{C}$.  Let $\varphi$ be a QMS on $\mathcal{A}_V$ together with its sequence of localized transition expectations $\mathcal{E}_{W_n} = \bigotimes_{x\in W_n}\mathcal{E}_x$, here $\mathcal{E}_x$ is a transition expectation from $\mathcal{A}_{\{x\}\cup S(x)}$ into $\mathcal{A}_x$. Put $\mathcal{B}_x = \mathcal{E}_x(\mathcal{A}_{\{x\}\cup S(x)})\subseteq\mathcal{A}_x$.

If the QMS $\varphi$ is translation-invariant, then the transition expectations $\mathcal{E}_x$ can be considered as copies of a transition expectation $\mathcal{E}$ from $M_d^{\otimes (k+1)}$ into $M_d$. In this setting, we denote $\mathcal{B}= \mathcal{E}(M_d^{\otimes (k+1)})$.

Now, let $p_1,\dots,p_r$ be the minimal central projections (counted with their multiplicities) of $\mathcal{B}$, so that the reduced algebras $p_i\mathcal{B}p_i\cong M_{d_i}$ for some $d_i\in\mathbb{N}$. Then
$$
\mathcal{B} = \bigoplus_{i=1}^{r}p_i\mathcal{B}p_i= \bigoplus_{i=1}^{r}M_{d_i}.
$$
Define
$$
E(A) = \sum_{i=1}^{r} p_iAp_i
$$
Put $\mathbb{T} = \{z\in\mathbb C:\quad |z|=1 \}.$

\begin{definition}
A  QMS $\varphi$ on $\mathcal{A}$  is said to be \textit{strongly mixing} if it satisfies
$$
\lim_{|g|\to\infty} \varphi(a \alpha_g(b)) = \varphi(a)\varphi(b)
$$
for all $a,b\in\mathcal{A}_V$.
\end{definition}

Any translation-invariant QMS $\varphi$ on $\mathcal{A}_V$ satisfies
\begin{equation}
\varphi\lceil_{\mathcal{A}_{V_g}}\circ\, \alpha_g  = \varphi
\end{equation}
where $V_g = \alpha_g(V)$.

\begin{theorem}\label{thm_mixing}
Every faithful, translation invariant QMS with localized transition expectations on $\mathcal{A}_V$ is strongly mixing.
\end{theorem}
\begin{proof}
Let $\varphi$ be a  translation-invariant  QMS on $\mathcal{A}_V$ together with its sequence of localized transition  expectations $\mathcal{E}_{W_n}=\bigotimes_{x\in  W_n}\mathcal{E}_{x}$ as in (\ref{localizedTE}), here $\mathcal{E}_x\equiv \mathcal{E}$ is a Umegaki conditional expectation from $\mathcal{A}_{\{x\}\cup S(x)}$ into $\mathcal{A}_x$. For each $x\in W_n$ and $j\in\{1,2,\dots, k\}$,  we put
$$
P_{j}= P_{x; j}:= \mathcal{E}_x\lceil_{\mathcal{A}_x\otimes \mathcal{B}^{(x,j)}\otimes\id_{S(x)\setminus \{(x,j)\}}}.
$$
One can see that $P_{x; j}$ is completely positive and identity preserving from $\mathcal{B}$ into itself. Let $\lambda\in\mathbb{T}$ and $B\in \mathcal{B}\setminus \{0\}$ such that $P_{x;j}(B^{(x,j)})= \lambda B^{(x)}$. Then for every $A' \in \mathcal{B}$,
\begin{eqnarray*}
\lambda (B'B)^{(x)} &=& (B'^{(x)}\otimes \id_{S(x)}) P_{x;j}(B^{(x,j)})\\
&=&  (B'^{(x)}\otimes \id_{S(x)}) \mathcal{E}_x(\id\otimes B^{(x,j)})\\ &=&  \mathcal{E}_x(B'^{(x)}\otimes\id_{S(x)\setminus \{(x,j)\}}\otimes B^{(x,j)})\\
&=& \mathcal{E}_x( \id_{\{x\}\cup S(x)\setminus }\otimes B^{(x,j)})B'^{(x)}\otimes\id_{S(x)} \\
&=& \lambda (BB')^{(x)}\otimes \id_{S(x)}.
\end{eqnarray*}
Then $B$ belongs to the center of $\mathcal{B}$, and there exist $\beta_i\in\mathbb{C} (1\le i\le r$) such that $B = \sum_{i=1}^{r}\beta_ip_i$. From \cite{AcLib99, MS19, MS20} there exist positive linear functional $\rho_{ii'}$ on $M_i\otimes M_{i'}$ such that
\begin{equation}\label{rhoii'}
\mathcal{E}_x(p_i^{(x)}\otimes p_{i'}^{(x,j)}\otimes \id_{S(x)\setminus\{(x,j)\}}) = \rho_{ii'}(p_i\otimes p_{i'}) p_{i}
\end{equation}
One has
\begin{eqnarray*}
P_{x,j}(B) &=& \mathcal{E}_x(\id^{(x)}\otimes B^{(x,j)}\otimes\id_{S(x)\setminus\{(x,j)\}})\\
&=& \sum_{ii'}\mathcal{E}_x(p_{i}\otimes \beta_{i'}p_{i'} \otimes \id_{S(x)\setminus\{(x,j)\}})\\
&=& \sum_{ii'}\beta_{i'}\mathcal{E}_x(p_{i}\otimes p_{i'} \otimes \id_{S(x)\setminus\{(x,j)\}})\\
&\overset{(\ref{rhoii'})}{=}& \sum_{ii'}\beta_{i'}\pi_{ii'}p_{i}\\
\end{eqnarray*}
where $\beta_{ii'} = \rho_{ii'}(p_i\otimes p_{i'})>0$. We have
$$
\left[ \begin{array}{ccc}
\pi_{11} & \cdots & \pi_{1r}\\
\vdots & \ddots & \vdots\\
\pi_{r1} & \cdots & \pi_{rr}
\end{array}\right]
\left[ \begin{array}{lll}
\beta_{1} \\
\vdots \\
\beta_{r}
\end{array}\right] =
\lambda \left[ \begin{array}{lll}
\beta_{1} \\
\vdots \\
\beta_{r}
\end{array}\right]
$$
Since $\pi_{ii'}>0$ for all $ii'$ then by the Perron-Frobenius theorem the map $P_{x,j}$ has trivial peripheral spectrum. Therefore, by \cite{F.N.W} the QMS $\varphi$ is strongly mixing.
\end{proof}

Now, we are going to establish an analogue of \cite[Theorem 3.1]{Go98}.
 \begin{theorem}\label{main}
 Let $\varphi$ be a locally faithful QMS on $\mathcal{A}_V$. Let $(\mathcal{H}, \pi, \Omega )$ be its associated GNS triplet.
 The following assertions hold true.
 \begin{itemize}
 \item[(i)]  The quantum Markov state $\varphi$ on the quasi-local algebra $\mathcal{A}_V$ is separating.
 i.e. the cyclic vector $\Omega$ is cyclic for the algebra $\mathcal{B}:= \pi(\mathcal{A}_V)^{''}$.
 \item[(ii)] The algebra $\mathcal{B} =\pi(\mathcal{A}_V)^{''} $ is a factor.
 \end{itemize}
 \end{theorem}

 \begin{proof} (i).
  Let $(\mathcal{H}, \pi, \Omega)$ be the GNS-triplet associated with the state locally faithful state $\varphi$.
Let $a\in\mathcal{A}_{\Lambda_{n}}$, we define an operator  $\hat{a}$ as
 $$
\hat{a}\pi(b)\Omega = \pi(ba)\Omega; \quad \forall b\in \mathcal{A}_{L, loc}.
 $$
   The operator $\hat{a}$ acts on $\pi(\mathcal{A}_{L})$ and commutes with $\pi(b)$ for every $b\in \mathcal{A}_{L, loc}$.

    On the other hand, we have
   \begin{eqnarray*}
   \langle \hat{a} \pi(b) \Omega,  \hat{a} \pi(b) \Omega \rangle
   &=&\varphi(a^\ast b^\ast  ba ) \\
   &=& \varphi(E_{n+1]}(a^\ast b^\ast  ba))\\
   &=& \varphi(a^\ast E_{n+1]}( b^\ast b)a)\\
    &=& \Tr\left((D_{\Lambda_{n+1}}a^\ast E_{n+1]}( b^\ast b)a \right)\\
   &=& \Tr\left(D_{\Lambda_{n+1}}^{\varphi})^{1/2}
    E_{n+1]}( b^\ast b)(D_{\Lambda_{n+1}}^{\varphi})^{1/2}\Big|(D_{\Lambda_{n+1}}^{\varphi})^{-1/2}a(D_{\Lambda_{n+1}}^{\varphi})^{1/2}\Big|^2\right)\\
    &\le & \Tr\left((D_{\Lambda_{n+1}}^{\varphi})^{1/2}
    E_{n+1]}( b^\ast b)(D_{\Lambda_{n+1}}^{\varphi})^{1/2}\right)\Big\|(D_{\Lambda_{n+1}}^{\varphi})^{-1/2}a(D_{\Lambda_{n+1}}^{\varphi})^{1/2}\Big\|^2.
   \end{eqnarray*}
and thanks to (\ref{Markov_state_eq}), one get 
\begin{eqnarray*}
\Tr\left((D_{\Lambda_{n+1}}^{\varphi})^{1/2}
    E_{n+1]}( b^\ast b)(D_{\Lambda_{n+1}}^{\varphi})^{1/2}\right)& =& \Tr\left( D_{\Lambda_{n+1}}^{\varphi}
    E_{n+1]}( b^\ast b)\right)\\
&=& \varphi( E_{n+1]}( b^\ast b))\\
&=& \varphi(b^\ast b)\\
&=& \| \pi(b) \Omega\|^2.
\end{eqnarray*}
Thus
$$
\|\hat{a}\| \le \Big\|(D_{\Lambda_{n+1}}^{\varphi})^{-1/2}a(D_{\Lambda_{n+1}}^{\varphi})^{1/2}\Big\|.
$$
It follows that $\hat{a}\in \pi(\mathcal{A}_L)'$. The set $\{\hat{a}\Omega= \pi(a) \Omega : a\in \mathcal{A}_{L, loc} \}$ is dense in $\mathcal{H}$. This means that
$\Omega$ is cyclic for $\pi(\mathcal{A}_L)^{'}$, equivalently $\Omega$ is separating for $\pi(\mathcal{A}_L)^{''}$. This proves (i).

Now according to Theorem \ref{thm_mixing}, from \cite{BR} we arrive at (ii).
%
%
 \end{proof}


\section*{Credit Author Statement}
The all authors of this paper are equally contributed for the realization of the results.

\section*{Declaration of Competing Interest}

The authors confirm that there are no known conflicts of interest associated with this publication and there has been no significant financial support for this work that could have influenced its outcome.

\section*{Data availability}
The paper does not use any data.

\section*{Acknowledgments}

The authors gratefully acknowledge Qassim University, represented by the Deanship of Scientific Research,
 on the financial support for this research under the number (10173-cba-2020-1-3-I) during the academic year 1442 AH / 2020 AD.

\end{document}